\documentclass[12pt]{article}

\usepackage[hmargin=2.32cm, vmargin={3cm,3cm}, a4paper]{geometry}\usepackage{amsmath,amssymb,amsthm,mathrsfs}
\usepackage{epsfig,epsf,subfigure,graphicx,graphics}
\usepackage{url,enumerate}
\graphicspath{{fig/}}
\usepackage{float}
\usepackage{parskip}
\usepackage{fancyhdr}
\usepackage[dvipsnames]{xcolor}
\usepackage{physics}

\DeclareMathAlphabet\mathrsfso      {U}{rsfso}{m}{n}
\usepackage{amssymb}\usepackage{mathtools}
\usepackage[toc,page]{appendix}
\usepackage{bbm}

\usepackage{lettrine}
\usepackage{hyperref}

\newtheorem{lemma}{Lemma}
\newtheorem*{lemma*}{Lemma}
\newtheorem*{proposition*}{Proposition}

\newtheorem{definition}{Definition}
\newtheorem{proposition}{Proposition}
\newtheorem{theorem}{Theorem}
\newtheorem*{theorem*}{Theorem}

\begin{document}

\title{Analysis of the sample complexity for PAC-learning functions defined over quantum states}

\author{Jordi Pérez-Guijarro}
\date{}

\maketitle

\abstract{A fundamental question in PAC learning is determining the number of labeled examples required to learn a concept class to a desired accuracy and confidence. In classical learning theory, this quantity is characterized by the VC-dimension, while several quantum generalizations have established analogous results when examples are provided in quantum superposition. In this work, we study a distinct quantum PAC-learning model in which concepts are functions acting on quantum states. We demonstrate that the VC-dimension, although still relevant, fails to fully capture the sample complexity of this model. To further characterize this setting, we develop a new lower bound on the required number of samples and establish an upper bound when the states in the domain are linearly independent. Remarkably, this upper bound has a form similar to the classical PAC-learning bound. We further examine a setting in which the learner receives more informative data and show that the limitations of the VC-dimension persist in this extended model.}

\section{Introduction}

The Probably Approximately Correct (PAC) learning framework, introduced by Valiant \cite{valiant1984theory}, is widely used for analyzing learning problems. In this framework, examples are drawn independently from an unknown probability distribution. Given a finite set of labeled examples, the learner is required to produce a hypothesis that, with probability at least $1-\delta$, has generalization error at most $\epsilon$. This setting has been extended in several directions, including the study of computationally efficient learning algorithms \cite{blumer1989learnability}, learning in the presence of noise \cite{angluin1988learning}, and quantum generalizations \cite{arunachalam2017guest,bshouty1995learning}.

Independently of the variant studied, one of the main quantities of focus is the sample complexity. This quantity determines how many examples are required to learn a concept class $\mathcal{C}$. In particular, the sample complexity for the original framework is given by $t_{\mathcal{C}}(\epsilon,\delta)=\Theta\left(\frac{d}{\epsilon} + \frac{\ln(1/\delta)}{\epsilon}\right)$, where $d$ denotes the VC-dimension of the concept class $\mathcal{C}$.

This characterization emerges from a progressive refinement of bounds over time. Specifically, in \cite{blumer1989learnability}, a lower bound of the form $\Omega\left(d + \frac{\ln(1/\delta)}{\epsilon}\right)$ was established, which was shortly after improved to the final value in \cite{ehrenfeucht1989general}. On the upper bound side, \cite{vapnik2006estimation,blumer1989learnability} show that any consistent learner achieves $O\left(\frac{d\ln(1/\epsilon)}{\epsilon} + \frac{\ln(1/\delta)}{\epsilon}\right)$, which is tight up to a $\ln(1/\epsilon)$ factor. This gap was later reduced by \cite{simon2015almost}, using a majority vote over classifiers. Finally, using similar techniques, \cite{hanneke2016optimal} was able to close the gap completely.

Similarly, for the quantum generalization of the PAC setting introduced in \cite{bshouty1995learning}, where instead of classical pairs $(x,f(x))$ a superposition of states $\ket{x}\ket{f(x)}$ is used, the sample complexity coincides with the classical one \cite{arunachalam2018optimal}. That is, in both cases, the VC-dimension uniquely determines the sample complexity. However, when restricted to efficient learners, as shown in \cite{servedio2004equivalences}, there are cases where the quantum sample complexity is strictly smaller than the classical sample complexity. Furthermore, if we allow access to both the unitary $U$ that generates the quantum examples and $U^\dagger$, then a quantum separation in sample complexity exists without the need to restrict to efficient learners \cite{salmon2023provable}.

In this work, we study another quantum generalization in which functions are defined over a Hilbert space and examples are given by pairs $(\ket{\phi}, f(\ket{\phi}))$. This setting was introduced in \cite{aimeur2006machine,gambs2008quantum}, where several algorithms are proposed. The specific case in which the input space is restricted to only two possible quantum states has been studied in \cite{guctua2010quantum,sentis2012quantum} within similar frameworks. This generalization is closely related to the problem of learning a measurement that best fits the data \cite{heidari2021theoretical,heidari2024new,heidari2023learning}. The key distinction is that, in this setting, the objective is to learn a function that approximates $f$, without requiring it to correspond to a physically implementable measurement.

The rest of the paper is organized as follows. Section \ref{notation} introduces the notation and framework used. Section \ref{VC_dimension_insufficiency} shows that the VC-dimension is not sufficient to characterize the sample complexity in this setting. Next, in Section \ref{Bounds_on_sample_complexity} and Section \ref{upper_bound_section}, we derive and analyze upper and lower bounds on the sample complexity. Section \ref{discussion_examples} considers a variant of the setting in which the examples are of the form $(\ket{\phi}^{\otimes k}, f(\ket{\phi}))$. Finally, Section \ref{conclusions} presents the conclusions of this article.

\section{Notation and Definitions}
\label{notation}

We start by introducing the notation. Let $[n] := \{1, \cdots, n\}$, and let $2^{S}$ denote the power set of a set $S$, i.e., the set of all subsets of $S$. In the PAC setting, we consider the problem of learning a concept $c : \mathcal{X} \rightarrow \{0,1\}$ belonging to some concept class $\mathcal{C}$, where $\mathcal{X}$ denotes the instance space. Importantly, in some parts of the paper, we identify a concept $c$ with a subset of $\mathcal{X}$. In particular, this subset consists of the inputs on which the function equals $1$ and completely characterizes the function. To perform this learning task, we have access to an oracle that outputs pairs $(X,c(X))$, where $X\sim D$, for an unknown concept $c$ and unknown distribution $D$. The objective is to output a hypothesis $h:\mathcal{X}\rightarrow\{0,1\}$ such that the error
\begin{equation}
    \mathrm{dist}_D(h,c):=\mathbb{P}_{X\sim D} (h(X)\neq c(X))
\end{equation}
is small. The hypothesis $h$ is the output of a learning algorithm $\mathcal{A}$ whose input is a training set $\{(x_i,c(x_i))\}_{i=1}^t$.

We say that $\mathcal{A}$ is an $(\epsilon,\delta)$-PAC learner for a concept class $\mathcal{C}$ with sample complexity $t(\epsilon,\delta)$ if, for every concept $c\in\mathcal{C}$ and every distribution $D$ over $\mathcal{X}$, the hypothesis $h=\mathcal{A}(\{(x_i,c(x_i))\}_{i=1}^{t(\epsilon,\delta)})$ where the values $\{x_i\}_i$ are i.i.d.\ realizations of the random variable $X\sim D$, satisfies
\begin{equation}
   \mathbb{P}_{X\sim D}\bigl(h(X)\neq c(X)\bigr)\leq \epsilon.
\end{equation}
with probability at least $1-\delta$ over the internal randomness of the algorithm $\mathcal{A}$ and the random samples $\{(x_i,c(x_i))\}_{i=1}^{t(\epsilon,\delta)}$. The minimum value of $t(\epsilon,\delta)$ over all learning algorithms is called the sample complexity of the concept class $\mathcal{C}$, and we denote it by $t_{\mathcal{C}}(\epsilon,\delta)$. If the concept class is not $(\epsilon,\delta)$-PAC learnable, then $t_{\mathcal{C}}(\epsilon,\delta)$ is defined to be equal to $\infty$. 

Importantly, whether a concept class is PAC-learnable is completely characterized by its VC-dimension $d$. The VC-dimension is the maximum cardinality of a set $\{x_i\}_{i=1}^k$ such that, for every function $g:\{x_i\}_{i=1}^k\rightarrow\{0,1\}$, there exists a concept $c\in\mathcal{C}$ satisfying $c(x_i)=g(x_i)$ for all $i\in[k]$. A set $\{x_1,\ldots,x_k\}$ satisfying this property is said to be shattered by the concept class $\mathcal{C}$. Moreover, the sample complexity is also characterized by the VC-dimension. Specifically, for $\epsilon\leq 1/8$ and $\delta\leq 1/100$,
\begin{equation}\label{sample_complexity_classical}
    c_L\left(\frac{d+\ln\frac{1}{\delta}}{\epsilon} \right) \leq t_{\mathcal{C}}(\epsilon,\delta)\leq  c_U \left(\frac{d+\ln\frac{1}{\delta}}{\epsilon} \right)
\end{equation}
for some constants $c_L$ and $c_U$ that are independent of the concept class $\mathcal{C}$.

Now, we introduce the quantum generalization of the PAC learning framework studied in this work. Let $\mathcal{X}$ consist of quantum states $\ket{\psi}$ belonging to some Hilbert space $\mathcal{H}$. As before, we have access to an oracle that outputs examples. In this case, however, the examples are of the form $(\ket{\phi},c(\ket{\phi}))$, where $\ket{\phi}$ is drawn according to some unknown distribution $D$.

We say that a learning algorithm $\mathcal{A}$ is an $(\epsilon,\delta)$-PAC learner for a concept class $\mathcal{C}$ using quantum examples with sample complexity $t(\epsilon,\delta)$ if, for every concept $c\in\mathcal{C}$ and every distribution $D$ over $\mathcal{X}$, the hypothesis $h=\mathcal{A}\left(\{(\ket{\phi_i},c(\ket{\phi_i}))\}_{i=1}^{t(\epsilon,\delta)}\right)$
satisfies
\begin{equation}
   \mathbb{P}_{\ket{\psi}\sim D}\!\left(h(\ket{\psi})\neq c(\ket{\psi})\right)\leq \epsilon,
\end{equation}
with probability at least $1-\delta$ over the internal randomness of the algorithm, the random samples $\{(\ket{\phi_i},c(\ket{\phi_i}))\}_{i=1}^{t(\epsilon,\delta)}$, and the outcomes of the measurements. Similarly, we define the sample complexity when quantum examples are used as the minimum value of $t(\epsilon,\delta)$ over all learning algorithms, which we denote by $t_{Q,\mathcal{C}}(\epsilon,\delta)$.

Note that, in general, the classical setting can be seen as a special case of this framework, where the quantum instance space consists of a set of orthonormal states $\ket{x}$, each of which is associated with an element $x$ of the original instance space $\mathcal{X}$.

Finally, we introduce an argument that is used at several points in the article. In particular, we observe that the quantum sample complexity $t_{Q,\mathcal{C}}(\epsilon,\delta)$ can be lower bounded by a simple argument. Namely, consider an oracle that outputs pairs $(\mathrm{des}(\ket{\phi}),c(\ket{\phi}))$, where $\mathrm{des}(\ket{\phi})$ denotes a classical description of the quantum state $\ket{\phi}$. This is a classical oracle, and therefore the sample complexity is lower bounded by the expression in \eqref{sample_complexity_classical}. Next, using that the classical oracle is more informative than the quantum oracle, since $(\mathrm{des}(\ket{\phi}),c(\ket{\phi}))$ allows one to recover $(\ket{\phi},c(\ket{\phi}))$, we conclude that the lower bound also applies to the quantum sample complexity, i.e.,
\begin{equation}
    c_L\left(\frac{d+\ln\frac{1}{\delta}}{\epsilon} \right)\leq t_{Q,\mathcal{C}}(\epsilon,\delta)
\end{equation}
for $\epsilon\leq 1/8$ and $\delta\leq 1/100$. Consequently, only concept classes with finite VC-dimension are learnable using quantum examples.

\section{Insufficiency of VC Dimension for Characterizing Quantum Sample Complexity}
\label{VC_dimension_insufficiency}

In this first section, we address the question of whether, in this generalized setting, the VC-dimension also characterizes the sample complexity. Specifically, we find that, in general, the VC-dimension is not sufficient to completely characterize the quantum sample complexity. To obtain this result, we need to introduce a new concept:

\begin{definition}
    A pair of quantum states $\{\ket{\psi},\ket{\psi'}\}$ is an opposite pair with respect to a concept class $\mathcal{C}$ if there exist two concepts $c,c'\in \mathcal{C}$ such that $c(\ket{\psi})=c'( \ket{\psi'})=1$ and $c(\ket{\psi'})=c'( \ket{\psi})=0$.
\end{definition}

The existence of an opposite pair in a concept class is a rather mild condition, since any class with VC-dimension at least two satisfies it. Now that we have introduced the required concepts, we can formally state the result that the VC dimension is not sufficient to characterize the sample complexity.
\begin{proposition}\label{first_lower_bound}
If $\mathcal{C}$ has at least a pair of opposite states, then the quantum sample complexity $t_{Q,\mathcal{C}}(\epsilon,\delta)$ has to satisfy, for $\epsilon \leq 1/8$ and $\delta \leq 1/100$:
    \begin{equation}
        c_L\left(\frac{d}{\epsilon} + \ln\frac{1}{\delta} \left(\frac{1}{\epsilon} + \frac{1}{-\ln F}\right)\right) \leq t_{Q,\mathcal{C}}(\epsilon,\delta)
    \end{equation}
for some constant $c_L$ independent of $\mathcal{C}$, $F$ is the supremum of $|\bra{\psi}\ket{\psi'}|^2$ over all opposite pairs of quantum states $\{\ket{\psi}, \ket{\psi'}\}$ with respect to $\mathcal{C}$.
\end{proposition}

\begin{proof}
    See Appendix \ref{proof_lower_bound_F}.
\end{proof}

Therefore, the lower bound implies that even if the VC-dimension is finite, if $F=1$, then the concept class is not PAC-learnable. Note that $F=1$ is equivalent to the existence of a sequence of opposite pairs $\{\ket{\psi_i}, \ket{\psi_i'}\}_{i\in \mathbb{N}}$ such that $|\bra{\psi_i}\ket{\psi_i'}|^2 \geq 1-\Delta_i$, with $\Delta_i\rightarrow 0$ as $i\rightarrow \infty$. Hence, the condition $F=1$ implies that the domain $\mathcal{X}$ is not finite.

Another interesting aspect of this theorem is that it is tight for some concept classes. In particular, this holds for the concept class of VC-dimension 2 defined over states $\{\ket{\psi_1}, \ket{\psi_2}\}$.

\begin{proposition}
    The concept class $\mathcal{C}:= 2^{\mathcal{X}}$ for $\mathcal{X}=\{\ket{\psi_1},\ket{\psi_2}\}$ has quantum sample complexity
    \begin{equation}
        t_{Q,\mathcal{C}}(\epsilon,\delta)=\Theta\left(\ln\frac{1}{\delta} \left(\frac{1}{\epsilon} + \frac{1}{-\ln F}\right)\right)
    \end{equation}
\end{proposition}

\begin{proof}

    The lower bound is given in Proposition \ref{first_lower_bound}. For the upper bound, we specify an algorithm that achieves this scaling. The algorithm we use is the following:

    \begin{itemize}
        \item If all labels are equal to some value $b$, output the hypothesis $h = [b,b]$. Otherwise, perform a Helstrom measurement to discriminate between the states
        \begin{equation}
             |\psi_1\rangle^{\otimes N_0} \otimes |\psi_2\rangle^{\otimes t - N_0} \quad \quad \text{and} \quad |\psi_2\rangle^{\otimes N_0} \otimes |\psi_1\rangle^{\otimes t - N_0} \quad
        \end{equation}
        where $N_0$ denotes the number of states with label $0$. If the first outcome is obtained, the algorithm outputs $h = [0,1]$, otherwise, it outputs $h = [1,0]$.
    \end{itemize}
    Consider the case $h = [b,b]$ for $b \in \{0,1\}$. Since the output is consistent with the $t$ oracle calls, and using the same derivation as in Theorem 2.2 of \cite{blumer1989learnability}, it follows that for $t \geq \frac{\ln(4/\delta)}{\epsilon}$, the error is at most $\epsilon$ with probability at least $1 - \delta$.

    Now consider the case where the algorithm outputs $h = [0,1]$ or $h = [1,0]$. For $t = O\left(\frac{\ln(1/\delta')}{-\ln|\langle \psi_1 | \psi_2 \rangle|^2}\right)$, the correct hypothesis is output with probability at least $1 - \delta'$. This follows directly from the fact that the Helstrom measurement on $t$ copies has average error probability given by 
    \begin{equation}
        \frac{1}{2}\left(1- \sqrt{1- |\bra{\psi_1} \ket{\psi_2}|^{2t}}\right)
    \end{equation}
    Note that we have bounded the average probability of error for these two hypotheses by $\delta'$. This implies that the error for each hypothesis is at most $2\delta'$. Taking $\delta = 2\delta'$, the desired result follows.
\end{proof}

    These results seem to indicate that, although the quantum sample complexity is not characterized by the VC-dimension alone, it may be characterized by the pair $(d, F)$. This would imply that the difference between the classical setting and the quantum setting is merely an additional factor of $O\left(\frac{\ln(1/\delta)}{-\ln F}\right)$. However, this is not the case, as shown next.
    \begin{proposition}\label{unlernability_PAC}
        There exists a concept class $\mathcal{C}$ that is not PAC learnable using quantum examples and for which both the VC-dimension and $-\ln F$ are finite.  
    \end{proposition}

    \begin{proof}

        Let $\mathcal{X} := \{\ket{0}, \ket{1}, \ket{+}, \ket{-}\}$, and let $D$ be the uniform distribution over this domain. We show that the concept class $\mathcal{C} = 2^{\mathcal{X}}$, which contains the concepts $c_1 = \{\ket{0}, \ket{1}\}$ and $c_2 = \{\ket{+}, \ket{-}\}$, cannot be learned. Note that $F = 1/2$ and the VC dimension is $d = 4$, i.e., both quantities are finite. For the proof, we use the induced density matrix of the quantum samples for each concept,
        \begin{equation}
            \sigma_{D}^c:= \frac{1}{4} \sum_{i=1}^4 \ket{\psi_i}  \bra{\psi_i}\otimes \ket{c(\ket{\psi_i})} \bra{c(\ket{\psi_i})}
        \end{equation}
        Substituting the concepts $c_1$ and $c_2$ and the states, we obtain $\sigma_{D}^{c_1} = \sigma_{D}^{c_2} = \tau^{\otimes 2}$, where $\tau = I/2$ is the maximally mixed state. Consequently, any algorithm $\mathcal{A}$ will have the same distribution over hypotheses for both concepts $c_1$ and $c_2$. This implies that, for at least one of them, the PAC conditions are not satisfied. 
        
    \end{proof}

    Therefore, the pair $(d, F)$ is not sufficient to characterize the quantum sample complexity, leaving open which set of parameters is required for such a complete characterization.

    \section{Lower Bound on the Quantum Sample Complexity}
    \label{Bounds_on_sample_complexity}
    
    In this section, we improve the lower bound derived in Proposition \ref{first_lower_bound}. In particular, the new result incorporates the unlearnability of concept classes such as the one used in Proposition \ref{unlernability_PAC}. The bound uses some extra notation, in particular, the definition of the states
    \begin{equation}
        \sigma^{c}_{D} \;:=\;
        \mathbb{E}_{|\psi\rangle \sim D}
        \Big[\, |\psi\rangle\langle\psi| \otimes
        |c(|\psi\rangle)\rangle\langle c(|\psi\rangle)| \,\Big].
    \end{equation}
    The precise statement of the lower bound is given in Theorem \ref{theorem_lower_bound}.

    \begin{theorem}\label{theorem_lower_bound}
    Let $\mathcal{C}$ be any concept class defined over a set of quantum states. Then, for $\epsilon \leq 1/8$ and $\delta \leq 1/100$,
    \begin{equation}
        c_L\left(\frac{d}{\epsilon}+\ln\frac{1}{\delta} \left(\frac{1}{\epsilon} + \frac{1}{\Lambda_{\mathcal{C} } (3\epsilon)}\right)\right) \leq t_{Q,\mathcal{C}}(\epsilon,\delta)
    \end{equation}
    for some constant $c_L$ independent of the concept class $\mathcal{C}$. The function $\Lambda_{\mathcal{C}}$ is defined by
    \begin{equation}
        \Lambda_{\mathcal{C}}(\alpha):= \inf_{\substack{D,D',c,c':\\ \mathrm{dist}_{D,D'}(c,c')> \,\alpha \\ d_{TV}(D,D')\leq \frac{\alpha}{6}}} C_Q (\sigma_D^c, \sigma_{D'}^{c'})
    \end{equation}
    where $\mathrm{dist}_{D,D'}(c,c'):= \min\{\mathrm{dist}_{D}(c,c'),\mathrm{dist}_{D'}(c,c')\} $, $d_{TV}(D,D')$ denotes the total variation distance, and 
    \begin{equation}
        C_Q (\sigma, \rho):= -\ln \min_{s\in [0,1]} \Tr (\sigma^{s} \rho^{1-s})
    \end{equation}
    is the quantum Chernoff distance.
    \end{theorem}

    \begin{proof}
        See Appendix \ref{appendix_lower_bound}.
    \end{proof}

    Note that, in the scenario described in Proposition \ref{unlernability_PAC}, if there exist two concepts such that $\mathrm{dist}_{D}(c,c')=\alpha_0$ and $\sigma_{D}^{c}=\sigma_{D}^{c'}$ under some distribution $D$, then $\Lambda_{\mathcal{C}}(\alpha)=0$ for all $\alpha\leq\alpha_0$. This implies that, for all $\epsilon<\alpha_0/3$, the concept class is not PAC learnable, or equivalently, its quantum sample complexity is infinite. 

    The main difference between this lower bound and the one introduced in Proposition \ref{first_lower_bound} is the appearance of the function $\Lambda_{\mathcal{C}}(\alpha)$ in place of the constant $-\ln F$. This naturally raises the question of how these two quantities are related. In fact, they are closely related. For instance, this relationship can be seen by considering $\mathcal{C}:=2^{\mathcal{X}}$, where $\mathcal{X}=\{\ket{\psi_1},\ket{\psi_2}\}$.

    \begin{proposition}\label{proposition_example}
        For the concept class $\mathcal{C}:= 2^{\mathcal{X}}$, where $\mathcal{X}=\{\ket{\psi_0},\ket{\psi_1}\}$,
        \begin{equation}
            \min \{-\ln F , \alpha (1-F)\}\leq \Lambda_{\mathcal{C}}(\alpha) \leq   -\ln F
        \end{equation}
    \end{proposition}
    \begin{proof}
        See Appendix \ref{proof_of_proposition_4}.
    \end{proof}
    
    Interestingly, this relation is not specific to this concept class but is quite general, as shown below:

    \begin{proposition}\label{upper_bound_proposition}
        For any concept class $\mathcal{C}$ with VC-dimension $d \geq 2$ and for $\alpha \leq \frac{1}{2}$,
        \begin{equation}
            \Lambda_{\mathcal{C}}(\alpha) \leq \min \{-\ln F, 2\alpha\} 
        \end{equation}
    \end{proposition}
    \begin{proof}
        The part of the upper bound corresponding to $-\ln F$ can be derived using Proposition \ref{proposition_example}. For the other part, since $d \geq 2$, there exist concepts $c, c'$ and a pair of states $\ket{\psi_1}, \ket{\psi_0}$ such that
        \begin{equation}
            c(\ket{\psi_0})\neq c'(\ket{\psi_0}) \text{ and } c(\ket{\psi_1})= c'(\ket{\psi_1})
        \end{equation}
        Consider $D=D'$ with support $\{\ket{\psi_0},\ket{\psi_1}\}$, such that $\mathbb{P}(\ket{\psi_0})=\alpha+\eta$ and $\mathbb{P}(\ket{\psi_1})=1-\alpha-\eta$, then
        \begin{align}
            \Lambda_{\mathcal{C}}(\alpha) &\leq \inf_{\eta>0 }C_Q(\sigma_D^c,\sigma_{D}^{c'}) \leq  \inf_{\eta>0 }\|\sigma_D^c- \sigma_{D}^{c'}\|_1 \nonumber \\& = \inf_{\eta>0 }(\alpha+\eta)  \| \ket{\psi_0}\bra{\psi_0}\otimes \ket{0}\bra{0}- \ket{\psi_0}\bra{\psi_0}\otimes \ket{1}\bra{1} \|_1 = 2\alpha 
        \end{align}
        where the second inequality will be proved next. We will use the fact that for states $\rho, \sigma$ such that $\|\rho-\sigma\|_1\leq 1$, it holds that 
        \begin{equation}
            C_Q(\rho,\sigma) \leq  \|\rho- \sigma\|_1
        \end{equation}
        Using expression (6) from \cite{audenaert2007discriminating},  
        \begin{equation}
            -\ln \inf_{s\in [0,1]} \Tr(\sigma^s \rho ^{1-s} )\leq -\ln \left( 1- \frac{\|\rho- \sigma\|_1}{2}\right)
        \end{equation}
        Consequently, since $\ln (1+x) \geq \frac{x}{1+x}$ for $x>-1$ \cite{topsoe2007some}, 
        \begin{equation}
            -\ln \inf_{s\in [0,1]} \Tr(\sigma^s \rho ^{1-s} )\leq  \frac{\|\rho- \sigma\|_1}{2-\|\rho- \sigma\|_1} \leq \|\rho- \sigma\|_1
        \end{equation}
        where the last inequality follows from the assumption that $\|\rho - \sigma\|_1 \leq 1$.
    \end{proof}

     \begin{figure}[H]
    	\centering
    	\includegraphics[width=13cm]{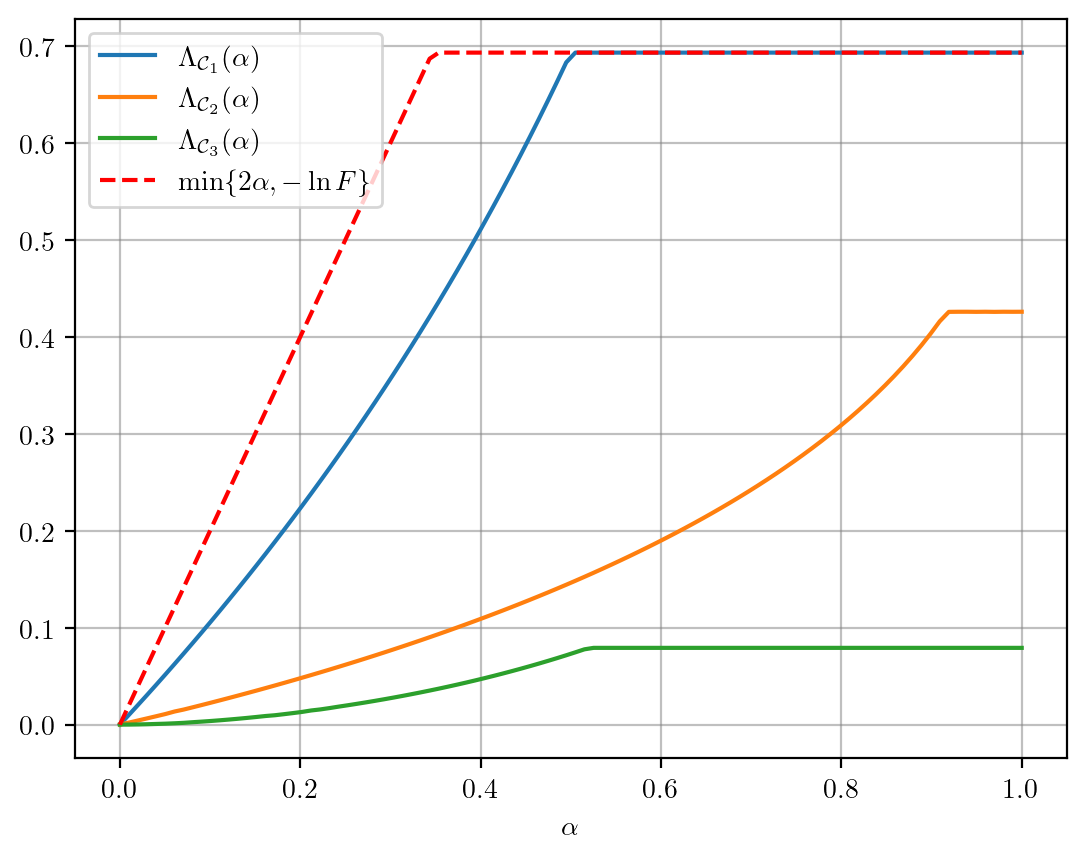}
    	\caption{Plot of the function $\Lambda_\mathcal{C}$ for the concept classes $\mathcal{C}_1=2^{\{\ket{0},\ket{+}\}}$, $\mathcal{C}_2=2^{\{\ket{0},\ket{+},\ket{+y}\}}$, and $\mathcal{C}_3=2^{\{\ket{0},\ket{1},\ket{+},\ket{+y}\}}$, along with the upper bound $\min\{-\ln F, 2\alpha\}$. The code used to generate this figure is available at \cite{perez2025code}.}
    	\label{Lambda_figure}
    \end{figure}

    Therefore, $\Lambda_{\mathcal{C}}(\alpha)$ is always bounded by the constant $-\ln F$. This can be seen clearly in Figure \ref{Lambda_figure}, where the function $\Lambda_\mathcal{C}$ is plotted for the concept classes $\mathcal{C}_1=2^{\{\ket{0},\ket{+}\}}$, $\mathcal{C}_2=2^{\{\ket{0},\ket{+},\ket{+y}\}}$, and $\mathcal{C}_3=2^{\{\ket{0},\ket{1},\ket{+},\ket{+y}\}}$. However, the plot suggests that $\Lambda_{\mathcal{C}}(\alpha)$ is, in general, upper bounded by some constant $B \leq -\ln F$. Indeed, this is the case. To define this constant, we first generalize the notion of an opposite pair of states as follows:

    \begin{definition}
        A pair of quantum density matrices $\{\sigma, \sigma'\}$ is an opposite pair with respect to a concept class $\mathcal{C}$ if there exist sets of states $\{\ket{\psi_i}\}_i$ and $\{|\psi'_j\rangle\}_j$ and concepts $c$ and $c'$ such that for all $i,j$,
        \begin{equation}
            c(\ket{\psi_i})=c'(\ket{\psi'_j})=1 \text{ and } c(\ket{\psi'_j})=c'(\ket{\psi_i})=0
        \end{equation}
        and
        \begin{equation}
            \sigma = \sum_i \alpha_i \ket{\psi_i}\bra{\psi_i} \text{ and } \sigma' = \sum_j \beta_j \ket{\psi'_j}\bra{\psi'_j}
        \end{equation}
        where $0\leq \alpha_i,\beta_i$, $\sum_i \alpha_i =1$ and $\sum_j \beta_j =1 $.
        
    \end{definition}

    We now introduce the constants suggested in Figure \ref{Lambda_figure}.
    
    \begin{proposition}
        Let $\sigma,\sigma'$ be opposite pairs with respect to a concept class $\mathcal{C}$, and let us denote by $\rho(\eta):= \eta\sigma \otimes \ket{1}\bra{1}+(1-\eta)\sigma' \otimes \ket{0}\bra{0}$ and $\rho'(\eta):=\eta\sigma \otimes \ket{0}\bra{0}+(1-\eta)\sigma' \otimes \ket{1}\bra{1}$, then
        \begin{equation}
            B_{\mathcal{C}}=\inf_{\sigma,\sigma',\eta} C_Q ( \rho(\eta),\rho'(\eta))
        \end{equation}
        satisfies that $\Lambda_\mathcal{C}(\alpha)\leq B_{\mathcal{C}}$.
    \end{proposition}

    The proof follows immediately from the definition of $\Lambda_\mathcal{C}(\alpha)$ and the notion of an opposite pair.
    
    Importantly, similarly to the case of the pair of constants $(d, F)$, the pair $(d, B_{\mathcal{C}})$ does not suffice to characterize the sample complexity. This is formally stated in the following proposition.

    \begin{proposition}\label{Proposition_7}
        There exists a concept class $\mathcal{C}$ that is learnable using quantum examples such that
        \begin{equation}
            \Lambda_\mathcal{C}(\alpha)\leq \alpha^2.
        \end{equation}
        It follows that the sample complexity of this concept class is lower bounded by $\Omega(1/\epsilon^2)$.
    \end{proposition}

    \begin{proof}
        See Appendix \ref{proof_of_proposition_7}.
    \end{proof}

    The condition that the concept class is learnable is included to exclude trivial concept classes, such as those satisfying $\Lambda_{\mathcal{C}}(\alpha)=0$, which are not learnable and would otherwise also satisfy the statement of Proposition \ref{Proposition_7}.

    This proposition implies that the function $\Lambda_{\mathcal{C}}(\alpha)$ cannot be characterized using only the constant $B_\mathcal{C}$, as the rate of change of the function can significantly impact the sample complexity and should be taken into account.

    \section{Upper Bound on the Quantum Sample Complexity}
    \label{upper_bound_section}

    Now, we turn our attention to upper bounds on the quantum sample complexity. We are going to analyze a scenario that is closely related to the classical setting and should therefore exhibit similar behavior. Specifically, we consider the case in which the states in the domain are linearly independent. 
    \begin{theorem}\label{proposition_sufficient}
        Let $\mathcal{X}:=\{\ket{\psi_i}\}_{i}$, where the states are linearly independent. Then any concept class $\mathcal{C}\subseteq 2^{\mathcal{X}}$ is PAC learnable from quantum examples. Moreover, its sample complexity is bounded by
        \begin{equation}
            t_{Q,\mathcal{C}}(\epsilon,\delta) \leq c_U \left(\frac{d+\ln\frac{1}{\delta}}{\epsilon \, \lambda_{\mathrm{min}}(G)}  \right) 
        \end{equation}
        where $G_{i,j}=\bra{\psi_i}\ket{\psi_j}$.
    \end{theorem}

    \begin{proof}
        See Appendix \ref{proof_proposition}.
    \end{proof}

   As expected, the sample complexity can be upper bounded by a quantity that closely resembles its classical counterpart. Indeed, the case $\lambda_{\mathrm{min}}(G)=1$ is precisely the classical one, as it corresponds to the states being mutually orthogonal. In this case, we recover the classical sample complexity. More generally, provided that $\lambda_{\mathrm{min}}(G)>0$, the sample complexity differs from the classical bound by at most a multiplicative factor of $1/\lambda_{\mathrm{min}}(G)$. Another important consequence of this result is that $\Lambda_{\mathcal{C}}(3\epsilon)=\Omega(\epsilon\lambda_{\min}(G))$. Hence, whenever $\lambda_{\min}(G)>0$, the function $\Lambda_{\mathcal{C}}(\alpha)$ is sandwiched between linear upper and lower bounds, ruling out behaviors such as that of Proposition \ref{Proposition_7}. Besides this connection, this quantity is also closely related to $-\ln F$. In particular, it satisfies the following inequality:
   \begin{equation}
        \lambda_{\mathrm{min}}(G)\leq 1- \max_{i,j: i\neq j} |\bra{\psi_i}\ket{\psi_j}|\leq -\ln \max_{i,j: i\neq j}|\bra{\psi_i}\ket{\psi_j}|  \leq -\frac{1}{2} \ln  F
   \end{equation}
    The proof of this inequality is shown in Appendix \ref{section_lambda}.

    \section{Discussion of More Informative Examples}
    \label{discussion_examples}

    So far, we have considered samples of the form $\{\ket{\psi}, c(\ket{\psi})\}$. In some cases, however, one may have access to multiple copies of the state, namely $\{\ket{\psi}^{\otimes k}, c(\ket{\psi})\}$. While this provides strictly more informative samples, the previous results still extend to this setting whenever $k$ is a fixed constant independent of $\mathcal{C}$.
    
    A more interesting regime arises when $k$ is allowed to depend on the concept class $\mathcal{C}$, and even on $\epsilon$ and $\delta$. We can then generalize the notion of $(\epsilon,\delta)$-learnability using quantum examples as follows: there exist an algorithm $\mathcal{A}$ and functions $k(\epsilon,\delta)$ and $t(\epsilon,\delta)$ such that, given input $\{\ket{\phi_i}^{\otimes k(\epsilon,\delta)}, c(\ket{\phi_i})\}_{i=1}^{t(\epsilon,\delta)}$, the algorithm outputs a hypothesis $h$ such that
    \begin{equation}
        \mathbb{P}_{\ket{\psi}\sim D}( h(\ket{\psi})\neq c(\ket{\psi}))\leq \epsilon     
    \end{equation}
    with probability at least $1-\delta$.

    In this framework, the situation differs from previous results. In particular, if the VC-dimension of $\mathcal{C}$ is finite and the domain $\mathcal{X}$ is finite, then there exists a finite value of $k(\epsilon,\delta)$ such that the concept class is learnable. Note that this is not the case for $k=1$, as a counterexample is given in the proof of Proposition~\ref{unlernability_PAC}.

\begin{proposition}
    Let $\mathcal{C}$ be a concept class with VC-dimension $d$ over a finite domain $\mathcal{X}$. Then there exists an algorithm $\mathcal{A}$ such that, for
    \begin{equation}
        k= O\left(\frac{\ln |\mathcal{X}|}{-\ln \max_{i,j: i\neq j}|\bra{\psi_i}\ket{\psi_j}| }\right),
    \end{equation}
    and
    \begin{equation}
        t(\epsilon,\delta)=O\left(\frac{d+\ln\frac{1}{\delta}}{\epsilon}  \right),
    \end{equation}
    the algorithm, on input $\left\{\ket{\phi_i}^{\otimes k},\, c(\ket{\phi_i})\right\}_{i=1}^{t(\epsilon,\delta)}$, outputs a hypothesis $h$ such that
    \begin{equation}
        \mathbb{P}_{\ket{\psi}\sim D}( h(\ket{\psi})\neq c(\ket{\psi})\bigr)\leq \epsilon
    \end{equation}
    with probability at least $1-\delta$.
\end{proposition}

\begin{proof}
    This result follows trivially by combining two results. First, that 
    \begin{equation}\label{inequality_lambda}
        \lambda_{\min}( G^{(k)})\geq 1-(|\mathcal{X}|-1)  \max_{i,j: i\neq j}|\bra{\psi_i}\ket{\psi_j}|^k 
    \end{equation}
    where $G^{(k)}_{i,j}:= \bra{\psi_i} \ket{\psi_j}^k$. The inequality is obtained by applying Gershgorin’s circle theorem. Secondly, we apply the algorithm used in Theorem \ref{proposition_sufficient} to the pairs $(\ket{\phi}^{\otimes k},c(\ket{\phi}))$. This can be done by choosing, for instance,
    \begin{equation}
        k=\frac{\ln 2|\mathcal{X}|}{-\ln \max_{i,j: i\neq j}|\bra{\psi_i}\ket{\psi_j}| }
    \end{equation}
    so that the states $\ket{\psi_i}^{\otimes k}$ are linearly independent, as implied by \eqref{inequality_lambda}.
\end{proof}

This result suggests that for sufficiently large $k$, the sample complexity matches the classical value. However, this is not true. Even when $k$ is allowed to depend on $\epsilon$ and $\delta$, there exist concept classes that are not learnable despite having finite VC-dimension.

    \begin{proposition}
        There exist concept classes $\mathcal{C}$ with finite VC-dimension $d$ that cannot be PAC-learned using quantum examples, independently of the function $k:(0,1]^2 \rightarrow \mathbb{N}$ used.
    \end{proposition}

    \begin{proof}
        For this proof, we start by using the fact that there exist concept classes such that the VC-dimension is finite and $F=1$. For instance, any concept class $\mathcal{C}=\{c_1,c_2\}$ such that $c_1= S(\mathcal{H})\setminus c_2$ and $c_2 \subset S(\mathcal{H})$ satisfies that the VC-dimension is $1$ and $F=1$, where $S(\mathcal{H})$ denotes the set of pure states in the Hilbert space $\mathcal{H}$. That is, there exists a sequence of opposite pairs $\{\ket{\psi_i}, \ket{\psi_i'}\}_{i\in \mathbb{N}}$ such that $|\bra{\psi_i}\ket{\psi_i'}|^2 \geq 1-\Delta_i$, with $\Delta_i \rightarrow 0$ as $i\rightarrow \infty$. Next, using Proposition~\ref{first_lower_bound}, for any function $k:(0,1]^2 \rightarrow \mathbb{N}$, we have that, in order to satisfy the PAC conditions on distributions restricted to the opposite states ${\ket{\psi_i}, \ket{\psi_i'}}$, the sample complexity must satisfy
        \begin{equation}
            t(\epsilon,\delta)\geq  \Omega\left(\frac{\ln \frac{1}{\delta}}{-k(\epsilon,\delta)\ln |\bra{\psi_i}\ket{\psi_i'}|^2}\right)
        \end{equation}
        Therefore, the sample complexity can be made arbitrarily large by taking $i$ sufficiently large, for any finite function $k(\epsilon,\delta)$.
    \end{proof}

    This is quite surprising, since even for large $k(\epsilon,\delta)$ there remains a substantial difference compared with the case where a classical description of the states is available (corresponding to $k=\infty$). In that regime, finite VC-dimension characterizes learnability of a concept class, whereas in our setting it is no longer sufficient.

    These discrepancies with the classical case arise from the fact that we require the algorithm to work for any distribution. Therefore, we can consider distributions whose support is concentrated on the boundaries of the concept classes, which implies that small errors due to the use of a finite number of copies $k$ significantly affect the performance of any algorithm. Consequently, such unlearnability results are expected to be rarer when considering variants of the PAC-learning setting in which the class of distributions is restricted.

    \section{Conclusions}
    \label{conclusions}

    In summary, we have shown that when using quantum examples ${(\ket{\phi},c(\ket{\phi}))}$ for PAC learning a concept class, the VC-dimension $d$ does not completely characterize the sample complexity, and therefore new quantities must be introduced. To characterize the quantum sample complexity, we first introduce a lower bound that captures some of the effects of having quantum examples instead of classical ones, via the function $\Lambda_{\mathcal{C}}(\alpha)$. This function captures the cases discussed in this paper where a concept class is not learnable, taking the value $\Lambda_{\mathcal{C}}(\alpha)=0$ for $\alpha>0$. We have also shown an upper bound on the quantum sample complexity when the states in the domain are linearly independent, which yields a similar result to the classical setting. Finally, we have discussed the scenario where multiple copies of the state are given in each quantum example, and, as before, we find that the VC-dimension $d$ does not fully characterize the sample complexity.
\bibliographystyle{ieeetr}
\bibliography{bibliography.bib}

\begin{thebibliography}{10}

\bibitem{valiant1984theory}
L.~G. Valiant, ``A theory of the learnable,'' {\em Communications of the ACM}, vol.~27, no.~11, pp.~1134--1142, 1984.

\bibitem{blumer1989learnability}
A.~Blumer, A.~Ehrenfeucht, D.~Haussler, and M.~K. Warmuth, ``Learnability and the vapnik-chervonenkis dimension,'' {\em Journal of the ACM (JACM)}, vol.~36, no.~4, pp.~929--965, 1989.

\bibitem{angluin1988learning}
D.~Angluin and P.~Laird, ``Learning from noisy examples,'' {\em Machine learning}, vol.~2, no.~4, pp.~343--370, 1988.

\bibitem{arunachalam2017guest}
S.~Arunachalam and R.~De~Wolf, ``Guest column: A survey of quantum learning theory,'' {\em ACM Sigact News}, vol.~48, no.~2, pp.~41--67, 2017.

\bibitem{bshouty1995learning}
N.~H. Bshouty and J.~C. Jackson, ``Learning dnf over the uniform distribution using a quantum example oracle,'' in {\em Proceedings of the eighth annual conference on Computational learning theory}, pp.~118--127, 1995.

\bibitem{ehrenfeucht1989general}
A.~Ehrenfeucht, D.~Haussler, M.~Kearns, and L.~Valiant, ``A general lower bound on the number of examples needed for learning,'' {\em Information and Computation}, vol.~82, no.~3, pp.~247--261, 1989.

\bibitem{vapnik2006estimation}
V.~Vapnik, {\em Estimation of dependences based on empirical data}.
\newblock Springer Science \& Business Media, 2006.

\bibitem{simon2015almost}
H.~U. Simon, ``An almost optimal pac algorithm,'' in {\em Conference on Learning Theory}, pp.~1552--1563, PMLR, 2015.

\bibitem{hanneke2016optimal}
S.~Hanneke, ``The optimal sample complexity of pac learning,'' {\em Journal of Machine Learning Research}, vol.~17, no.~38, pp.~1--15, 2016.

\bibitem{arunachalam2018optimal}
S.~Arunachalam and R.~De~Wolf, ``Optimal quantum sample complexity of learning algorithms,'' {\em The Journal of Machine Learning Research}, vol.~19, no.~1, pp.~2879--2878, 2018.

\bibitem{servedio2004equivalences}
R.~A. Servedio and S.~J. Gortler, ``Equivalences and separations between quantum and classical learnability,'' {\em SIAM Journal on Computing}, vol.~33, no.~5, pp.~1067--1092, 2004.

\bibitem{salmon2023provable}
W.~Salmon, S.~Strelchuk, and T.~Gur, ``Provable advantage in quantum pac learning,'' {\em arXiv preprint arXiv:2309.10887}, 2023.

\bibitem{aimeur2006machine}
E.~A{\"\i}meur, G.~Brassard, and S.~Gambs, ``Machine learning in a quantum world,'' in {\em Conference of the Canadian society for computational studies of intelligence}, pp.~431--442, Springer, 2006.

\bibitem{gambs2008quantum}
S.~Gambs, ``Quantum classification,'' {\em arXiv preprint arXiv:0809.0444}, 2008.

\bibitem{guctua2010quantum}
M.~Gu{\c{t}}{\u{a}} and W.~Kot{\l}owski, ``Quantum learning: asymptotically optimal classification of qubit states,'' {\em New Journal of Physics}, vol.~12, no.~12, p.~123032, 2010.

\bibitem{sentis2012quantum}
G.~Sent{\'\i}s, J.~Calsamiglia, R.~Munoz-Tapia, and E.~Bagan, ``Quantum learning without quantum memory,'' {\em Scientific reports}, vol.~2, no.~1, p.~708, 2012.

\bibitem{heidari2021theoretical}
M.~Heidari, A.~Padakandla, and W.~Szpankowski, ``A theoretical framework for learning from quantum data,'' in {\em 2021 IEEE international symposium on information theory (ISIT)}, pp.~1469--1474, IEEE, 2021.

\bibitem{heidari2024new}
M.~Heidari and W.~Szpankowski, ``New bounds on quantum sample complexity of measurement classes,'' in {\em 2024 IEEE International Symposium on Information Theory (ISIT)}, pp.~1515--1520, IEEE, 2024.

\bibitem{heidari2023learning}
M.~Heidari and W.~Szpankowski, ``Learning k-qubit quantum operators via pauli decomposition,'' in {\em International Conference on Artificial Intelligence and Statistics}, pp.~490--504, PMLR, 2023.

\bibitem{audenaert2007discriminating}
K.~M. Audenaert, J.~Calsamiglia, R.~Munoz-Tapia, E.~Bagan, L.~Masanes, A.~Acin, and F.~Verstraete, ``Discriminating states: The quantum chernoff bound,'' {\em Physical review letters}, vol.~98, no.~16, p.~160501, 2007.

\bibitem{topsoe2007some}
F.~Tops{\o}e, ``Some bounds for the logarithmic function,'' {\em Inequality theory and applications}, vol.~4, no.~01, 2007.

\bibitem{perez2025code}
J.~Pérez-Guijarro, ``Code and data: Analysis of the sample complexity for pac-learning functions defined over quantum states.'' \url{https://github.com/Jordi-Perez-Guijarro/code-sample-complexity-for-PAC-learning-over-quantum-states}.
\newblock July 8, 2026.

\bibitem{chefles1998unambiguous}
A.~Chefles, ``Unambiguous discrimination between linearly independent quantum states,'' {\em Physics Letters A}, vol.~239, no.~6, pp.~339--347, 1998.

\bibitem{mitzenmacher2017probability}
M.~Mitzenmacher and E.~Upfal, {\em Probability and computing: Randomization and probabilistic techniques in algorithms and data analysis}.
\newblock Cambridge university press, 2017.

\end{thebibliography}

\begin{appendices}

    \section{Proof of Proposition \ref{first_lower_bound}}\label{proof_lower_bound_F}

    In this appendix, we provide the proof of Proposition \ref{first_lower_bound}, restated here for completeness.

    \begin{proposition*}
        If $\mathcal{C}$ has at least a pair of opposite states, then the quantum sample complexity $t_{Q,\mathcal{C}}(\epsilon,\delta)$ has to satisfy, for $\epsilon \leq 1/8$ and $\delta \leq 1/100$:
            \begin{equation}
                c_L\left(\frac{d}{\epsilon} + \ln\frac{1}{\delta} \left(\frac{1}{\epsilon} + \frac{1}{-\ln F}\right)\right) \leq t_{Q,\mathcal{C}}(\epsilon,\delta)
            \end{equation}
        for some constant $c_L$ independent of $\mathcal{C}$, $F$ is the supremum of $|\bra{\psi}\ket{\psi'}|^2$ over all opposite pairs of quantum states $\{\ket{\psi}, \ket{\psi'}\}$ with respect to $\mathcal{C}$.
    \end{proposition*}

    \begin{proof}

        The term
        \begin{equation}
            c_L\left(\frac{d+\ln\frac{1}{\delta}}{\epsilon}\right)
        \end{equation}
        follows from the classical lower bound. For the other term, let $\ket{\psi_1}$ and $\ket{\psi_2}$ be a pair of opposite states. Let $\mathcal{A}$ be a learning algorithm for $\mathcal{C}$ that makes $t(\epsilon,\delta)$ oracle queries. Consider the uniform distribution $D$ over $\{\ket{\psi_1},\ket{\psi_2}\}$. For this distribution and for $\epsilon<1/2$, the algorithm must output the correct concept with probability at least $1-\delta$. Consequently, it must distinguish between two concepts, $c_1$ and $c_2$, that satisfy the conditions of an opposite pair. That is,
        \begin{align}
            &c_1(\ket{\psi_1})=1 \hspace{1cm} c_1(\ket{\psi_2})=0 \nonumber  \\ &  c_2(\ket{\psi_1})=0 \hspace{1cm} c_2(\ket{\psi_2})=1
        \end{align}
        Note that, for a given training set, this is equivalent to distinguishing between the states 
        \begin{equation}
            \ket{\phi_1}=\ket{\psi_1}^{\otimes N_0} \otimes \ket{\psi_2}^{\otimes t-N_0} \text{ and }\ket{\phi_2}=\ket{\psi_2}^{\otimes N_0} \otimes \ket{\psi_1}^{\otimes t-N_0}
        \end{equation}
        where $N_0$ denotes the number of quantum examples labeled $0$. For this particular realization of the oracle queries, it follows that the average probability of error must satisfy
        \begin{align}
            \mathbb{P}(\mathcal{E}|N_0)&\geq \frac{1}{2}\left(1-\frac{1}{2} \|\ket{\phi_1}\bra{\phi_1}- \ket{\phi_2}\bra{\phi_2} \|_1\right)\nonumber \\ & =\frac{1}{2}\left(1- \sqrt{1- |\bra{\psi_1} \ket{\psi_2}|^{2t}}\right)
        \end{align}
        Therefore, 
        \begin{align}
            \delta \geq \mathbb{E}_{N_0}[\mathbb{P}(\mathcal{E}|N_0)]&\geq \frac{1}{2}\left(1- \sqrt{1- |\bra{\psi_1} \ket{\psi_2}|^{2t}}\right)
        \end{align}
        which implies that, for $\delta\leq \frac{1}{2}$,
        \begin{equation}
            \frac{\ln \frac{1}{4\delta(1-\delta)}}{-\ln |\bra{\psi_1} \ket{\psi_2}|^2}\leq t 
        \end{equation}
    \end{proof}

    \section{Proof of Theorem \ref{theorem_lower_bound}}\label{appendix_lower_bound}

    In this appendix, we provide the proof of Theorem \ref{theorem_lower_bound}, restated here for completeness.

    \begin{theorem*}
    Let $\mathcal{C}$ be any concept class defined over a set of quantum states. Then, for $\epsilon \leq 1/8$ and $\delta \leq 1/100$,
    \begin{equation}
        c_L\left(\frac{d}{\epsilon}+\ln\frac{1}{\delta} \left(\frac{1}{\epsilon} + \frac{1}{\Lambda_{\mathcal{C} } (3\epsilon)}\right)\right) \leq t_{Q,\mathcal{C}}(\epsilon,\delta)
    \end{equation}
    for some constant $c_L$ independent of the concept class $\mathcal{C}$. The function $\Lambda_{\mathcal{C}}$ is defined by
    \begin{equation}
        \Lambda_{\mathcal{C}}(\alpha):= \inf_{\substack{D,D',c,c':\\ \mathrm{dist}_{D,D'}(c,c')> \,\alpha \\ d_{TV}(D,D')\leq \frac{\alpha}{6}}} C_Q (\sigma_D^c, \sigma_{D'}^{c'})
    \end{equation}
    where $\mathrm{dist}_{D,D'}(c,c'):= \min\{\mathrm{dist}_{D}(c,c'),\mathrm{dist}_{D'}(c,c')\} $, $d_{TV}(D,D')$ denotes the total variation distance, and 
    \begin{equation}
        C_Q (\sigma, \rho):= -\ln \min_{s\in [0,1]} \Tr (\sigma^{s} \rho^{1-s})
    \end{equation}
    is the quantum Chernoff distance.
    \end{theorem*}
    
    \begin{proof}
        If the learning algorithm $\mathcal{A}$ learns $\mathcal{C}$ using $t(\epsilon,\delta)$ examples, then, with probability at least $1-\delta$, it outputs a hypothesis $h$ such that
        \begin{equation}
            \mathrm{dist}_{D}(h,c)\leq \epsilon,
        \end{equation}
        where $c$ is the ground-truth concept and $D$ is the underlying distribution. Similarly,
        \begin{equation}
            \mathrm{dist}_{D'}(h,c')\leq \epsilon
         \end{equation}
        when $c'$ is the ground-truth concept and $D'$ is the underlying distribution. Therefore, if $\mathrm{dist}_{D,D'}(c,c'):=\min\{\mathrm{dist}_{D}(c,c'),\mathrm{dist}_{D'}(c,c')\}>\,3\epsilon$, and $d_{TV}(D,D')\leq \epsilon/2$, then, when $c$ is the ground-truth concept and $D$ is the underlying distribution,
        \begin{align}
            & \mathrm{dist}_D(h,c)\leq \epsilon \hspace{2cm}\mathrm{dist}_D(h,c') > 2\epsilon \nonumber \\ & \mathrm{dist}_{D'}(h,c)\leq 2\epsilon \hspace{1.75cm} \mathrm{dist}_{D'}(h,c') > \epsilon 
        \end{align}
        hold with probability at least $1-\delta$. The second inequality follows from
        \begin{equation}
            3\epsilon<\mathrm{dist}_{D,D'}(c,c')\leq \mathrm{dist}_{D}(c,c') \leq \mathrm{dist}_{D}(c,h) + \mathrm{dist}_{D}(h,c') \leq \epsilon + \mathrm{dist}_{D}(h,c')
        \end{equation}
        The third and fourth inequalities are a consequence of the bound
        \begin{equation}
            \left|\mathrm{dist}_{D}(c_1,c_2)-\mathrm{dist}_{D'}(c_1,c_2)\right|\leq 2d_{TV}(D,D') 
        \end{equation}
        for any pair of concepts $c_1, c_2$. Similarly, when $c'$ is the ground-truth concept and $D'$ is the underlying distribution,
        \begin{align}
            & \mathrm{dist}_{D'}(h,c')\leq \epsilon \hspace{2cm} \mathrm{dist}_{D'}(h,c) > 2\epsilon \nonumber \\ & \mathrm{dist}_{D}(h,c')\leq 2\epsilon \hspace{1.9cm} \mathrm{dist}_{D}(h,c) > \epsilon 
        \end{align}
        Consequently, when $c$ is the ground-truth concept and $D$ the underlying distribution, with probability at least $1-\delta$,
        \begin{equation}
            \mathrm{dist}_{D,D'}(h,c)\leq \epsilon \text{ and } \mathrm{dist}_{D,D'}(h,c')> \epsilon
        \end{equation}
        Similarly, when $c'$ is the ground-truth concept and $D'$ the underlying distribution, with probability at least $1-\delta$,
        \begin{equation}
            \mathrm{dist}_{D,D'}(h,c')\leq \epsilon \text{ and } \mathrm{dist}_{D,D'}(h,c)> \epsilon
        \end{equation}
        Therefore, the algorithm $\mathcal{A}$ can be used to discriminate between the states ${\sigma_{D}^{c}}^{\otimes t}$ and ${\sigma_{D'}^{c'}}^{\otimes t}$ with probability of error at most $\delta$. Hence,
        \begin{align}
            \delta &\geq \frac{1}{2} \left( 1 -\frac{1}{2} \left\| (\sigma_{D}^{c}) ^{\otimes t}-(\sigma_{D'}^{c'}) ^{\otimes t} \right\|_1 \right) \nonumber \\  &\geq  \frac{1}{2} \left( 1- \sqrt{1- Q(\sigma_{D}^{c},\sigma_{D'}^{c'})^{2t}} \right) 
        \end{align}
        where $Q(\rho,\sigma ):= \min_{s\in [0,1]} \Tr (\sigma^{s} \rho^{1-s})$. The second inequality uses (6) from \cite{audenaert2007discriminating}. Consequently, for $\delta<1/2$,
        \begin{equation}\label{equation_final}
            \frac{\ln \frac{1}{4\delta(1-\delta)}}{-2 \ln Q(\sigma_{D}^c,\sigma_{D'}^{c'})}\leq t(\epsilon,\delta)
        \end{equation}
        Finally, since \eqref{equation_final} holds for all $D,D'$ and $c,c'$ satisfying $\mathrm{dist}_{D,D'}(c,c')>\,3\epsilon$, and $d_{TV}(D,D')\leq \epsilon/2$ then 
        \begin{equation}
            \frac{\ln \frac{1}{4\delta(1-\delta)}}{2\Lambda_{\mathcal{C}}(3\epsilon)}\leq t(\epsilon,\delta)
        \end{equation}
        for $\delta\leq\frac{1}{2}$.
    \end{proof}

    \section{Proof of Proposition \ref{proposition_example}}
    \label{proof_of_proposition_4}

        In this appendix, we provide the proof of Proposition \ref{proposition_example}, restated here for completeness.

        \begin{proposition*}
            For the concept class $\mathcal{C}:= 2^{\mathcal{X}}$, where $\mathcal{X}=\{\ket{\psi_0},\ket{\psi_1}\}$,
            \begin{equation}
                \min \{-\ln F , \alpha (1-F)\}\leq \Lambda_{\mathcal{C}}(\alpha) \leq   -\ln F
            \end{equation}
        \end{proposition*}

        \begin{proof}
        For convenience, we label the concepts in the class as $\mathcal{C}={c_0,c_1,c_2,c_3}$, where the corresponding concepts are $c_0=[0,0]$, $c_1=[1,1]$, $c_2=[0,1]$, and $c_3=[1,0]$. The first coordinate denotes the value $c(\ket{\psi_0})$, and the second coordinate denotes the value $c(\ket{\psi_1})$.

        For the lower bound, we use the fact that
        \begin{equation}
            \inf_{\substack{D,D',c,c':\\ \mathrm{dist}_{D,D'}(c,c')> \,\alpha}} C_Q (\sigma_D^{c}, \sigma_{D'}^{c'}) \leq \Lambda_{\mathcal{C}}(\alpha)
        \end{equation}
        Next, note that
        \begin{equation}
            \inf_{\substack{D,D':\\ \mathrm{dist}_{D,D'}(c_0,c_1)> \,\alpha}} C_Q (\sigma_D^{c_0}, \sigma_{D'}^{c_1})= \infty
        \end{equation}
        as the states are orthogonal. Now, let us consider the pair $c_0$ and $c_2$,
        \begin{align}
            \inf_{\substack{D,D':\\ \mathrm{dist}_{D,D'}(c_0,c_2)> \,\alpha}} C_Q (\sigma_D^{c_0}, \sigma_{D'}^{c_2})= \inf_{\substack{D,D':\\ \mathrm{dist}_{D,D'}(c_0,c_2)> \,\alpha}} -\ln \min_{s\in [0,1]}  \Tr ((\sigma_D^{c_0})^{s} (\sigma_{D'}^{c_2})^{1-s}) 
        \end{align}
        where
        \begin{equation}
            (\sigma_D^{c_0})^s=\Big( p_0 \ket{\psi_0} \bra{\psi_0} + (1-p_0) \ket{\psi_1} \bra{\psi_1}\Big)^s \otimes \ket{0} \bra{0}
        \end{equation}

        \begin{equation}
            (\sigma_{D'}^{c_2})^{1-s}=\tilde{p}_0^{1-s} \ket{\psi_0} \bra{\psi_0}\otimes \ket{0} \bra{0} + (1-\tilde{p}_0)^{1-s} \ket{\psi_1} \bra{\psi_1}\otimes \ket{1} \bra{1} 
        \end{equation}
        Therefore, 
        \begin{align}
            \Tr ((\sigma_D^{c_0})^{s} (\sigma_{D'}^{c_2})^{1-s})&= \tilde{p}_0^{1-s}   \bra{\psi_0} \Big( p_0 \ket{\psi_0} \bra{\psi_0} + (1-p_0) \ket{\psi_1} \bra{\psi_1}\Big)^s \ket{\psi_0} \nonumber \\& \leq \bra{\psi_0} \Big( p_0 \ket{\psi_0} \bra{\psi_0} + (1-p_0) \ket{\psi_1} \bra{\psi_1}\Big)^s \ket{\psi_0} 
        \end{align}
        Hence, taking the minimum over $s\in[0,1]$ on both sides,
        \begin{align}
            \min_{s\in [0,1]}  \Tr ((\sigma_D^{c_0})^{s} (\sigma_{D'}^{c_2})^{1-s})&\leq \min_{s\in [0,1]}\bra{\psi_0} \Big( p_0 \ket{\psi_0} \bra{\psi_0} + (1-p_0) \ket{\psi_1} \bra{\psi_1}\Big)^s \ket{\psi_0} \nonumber \\&= p_0 + (1-p_0) |\bra{\psi_0}\ket{\psi_1}|^2 \nonumber \\&= p_0(1- |\bra{\psi_0}\ket{\psi_1}|^2) + |\bra{\psi_0}\ket{\psi_1}|^2
        \end{align}
        Finally, taking the supremum over the distributions, namely over $0\leq p_0<1-\alpha$, yields
        \begin{align}
            \inf_{\substack{D,D':\\ \mathrm{dist}_{D,D'}(c_0,c_2)> \,\alpha}} C_Q (\sigma_D^{c_0}, \sigma_{D'}^{c_2})  & \geq -\ln \left( (1-\alpha) (1-|\bra{\psi_0}\ket{\psi_1}|^2)+|\bra{\psi_0}\ket{\psi_1}|^2\right) \nonumber \\ &\geq - \left( (1-\alpha) (1-|\bra{\psi_0}\ket{\psi_1}|^2)+|\bra{\psi_0}\ket{\psi_1}|^2\right)+1 \nonumber \\ & = \alpha (1- |\bra{\psi_0}\ket{\psi_1}|^2)
        \end{align}
        By symmetry, the same argument applies to the pairs $\{c_0,c_3\}$, $\{c_1,c_2\}$, and $\{c_1,c_3\}$. Therefore, the only remaining pair to consider is $\{c_2,c_3\}$. For this pair, 
        \begin{equation}
            (\sigma_D^{c_2})^{s}={p}_0^{s} \ket{\psi_0} \bra{\psi_0}\otimes \ket{0} \bra{0} + (1-{p}_0)^{s} \ket{\psi_1} \bra{\psi_1}\otimes \ket{1} \bra{1} 
        \end{equation}
        \begin{equation}
            (\sigma_{D'}^{c_3})^{1-s}=\tilde{p}_0^{1-s} \ket{\psi_0} \bra{\psi_0}\otimes \ket{1} \bra{1} + (1-\tilde{p}_0)^{1-s} \ket{\psi_1} \bra{\psi_1}\otimes \ket{0} \bra{0} 
        \end{equation}
        Therefore,
        \begin{equation}
            \Tr ((\sigma_D^{c_2})^{s} (\sigma_{D'}^{c_3})^{1-s})= \Big(p_0^s  (1-\tilde{p}_0)^{1-s}+ (1-{p}_0)^{s}\tilde{p}_0^{1-s}\Big) |\bra{\psi_0}\ket{\psi_1}|^2
        \end{equation}
        The maximum value of the multiplicative factor is attained by choosing $p_0=1-\tilde{p}_0=1/2$, which yields
        \begin{equation}\label{equality_fidelity}
            \inf_{\substack{D,D':\\ \mathrm{dist}_{D,D'}(c_2,c_3)> \,\alpha}} C_Q (\sigma_D^{c_2}, \sigma_{D'}^{c_3})= - \ln |\bra{\psi_0}\ket{\psi_1}|^2
        \end{equation}
        Hence, 
        \begin{equation}
            \Lambda_{\mathcal{C}}(\alpha)\geq \min \{\alpha (1-F),-\ln F\}
        \end{equation}
        Since the optimal distributions for the pair $\{c_2,c_3\}$ satisfy $d_{\mathrm{TV}}(D,D')\leq \alpha/6$, we obtain the upper bound
        \begin{equation}
            \Lambda_{\mathcal{C}}(\alpha)\leq -\ln F
        \end{equation}
    \end{proof}

    \section{Proof of Proposition \ref{Proposition_7}}
    \label{proof_of_proposition_7}

    In this appendix, we provide the proof of Proposition \ref{Proposition_7}, restated here for completeness.

    \begin{proposition*}
        There exists a concept class $\mathcal{C}$ that is learnable using quantum examples such that
        \begin{equation}
            \Lambda_\mathcal{C}(\alpha)\leq \alpha^2.
        \end{equation}
        It follows that the sample complexity of this concept class is lower bounded by $\Omega(1/\epsilon^2)$.
    \end{proposition*}

    \begin{proof}
        Let $\mathcal{C}:=\{c_1,c_2\}$ be a concept class defined on the domain $\{\ket{0},\ket{+},\ket{+y},\ket{1}\}$, where $c_1=\{\ket{+},\ket{1}\}$, $c_2=\{\ket{+y},\ket{1}\}$, and $\ket{+y}$ denotes the +1 eigenstate of the Pauli $Y$ matrix. Let $D=D'=[\frac{1-\alpha}{2},\frac{\alpha}{2},\frac{\alpha}{2}, \frac{1-\alpha}{2}]$, where the entries correspond to the probabilities assigned to $\ket{0},\ket{+},\ket{+y},\ket{1}$, respectively. Consequently,
        \begin{align}
            \sigma_{D}^{c_1}= \frac{1}{2}&\Big( (1-\alpha) \ket{0}\bra{0}\otimes \ket{0}\bra{0}+\alpha \ket{+}\bra{+}\otimes \ket{1}\bra{1}\nonumber \\ &+\alpha \ket{+y}\bra{+y}\otimes \ket{0}\bra{0}+ (1-\alpha) \ket{1}\bra{1}\otimes \ket{0}\bra{0}\Big)
        \end{align}
        \begin{align}
            \sigma_{D'}^{c_2}= \frac{1}{2}&\Big( (1-\alpha) \ket{0}\bra{0}\otimes \ket{0}\bra{0}+\alpha \ket{+}\bra{+}\otimes \ket{0}\bra{0}\nonumber \\ &+\alpha \ket{+y}\bra{+y}\otimes \ket{1}\bra{1}+ (1-\alpha) \ket{1}\bra{1}\otimes \ket{1}\bra{1}\Big)
        \end{align}
        Therefore,
        \begin{equation}
            \Lambda_{\mathcal{C}}(\alpha)\leq -\ln \inf_{s\in (0,1)} \Tr((\sigma_{D}^{c_1})^s (\sigma_{D'}^{c_2})^{1-s})
        \end{equation}
        For the moment, we focus on the term $\Tr((\sigma_{D}^{c_1})^s (\sigma_{D'}^{c_2})^{1-s})$,
        \begin{align}
            &\Tr((\sigma_{D}^{c_1})^s (\sigma_{D'}^{c_2})^{1-s}) = \frac{1}{2} \Tr \Big [\big( (1-\alpha) \ket{0}\bra{0}+\alpha \ket{+y}\bra{+y}\big)^s \big((1-\alpha) \ket{0}\bra{0}+\alpha \ket{+}\bra{+} \big)^{1-s} \Big] \nonumber \\ & \hspace{1cm}+ \frac{1}{2}\Tr \Big [\big( \alpha \ket{+}\bra{+}+(1-\alpha) \ket{1}\bra{1}\big)^s \big(\alpha \ket{+y}\bra{+y}+(1-\alpha) \ket{1}\bra{1} \big)^{1-s} \Big] 
        \end{align}
        Consider each of the two terms separately,
        \begin{equation}
            f(s):=\Tr \Big [\big( (1-\alpha) \ket{0}\bra{0}+\alpha \ket{+y}\bra{+y}\big)^s \big((1-\alpha) \ket{0}\bra{0}+\alpha \ket{+}\bra{+} \big)^{1-s} \Big]
        \end{equation}
        \begin{equation}
            g(s):=\Tr \Big [\big( \alpha \ket{+}\bra{+}+(1-\alpha) \ket{1}\bra{1}\big)^s \big(\alpha \ket{+y}\bra{+y}+(1-\alpha) \ket{1}\bra{1} \big)^{1-s} \Big]
        \end{equation}
        We first show that both terms are equal. Using the cyclic property of the trace and the identity $UA^sU^\dagger=(UAU^\dagger)^s$ for any unitary $U$,
        \begin{align}
            f(s):&=\Tr \Big [\big( (1-\alpha) U\ket{0}\bra{0}U^{\dagger}+\alpha U\ket{+y}\bra{+y}U^\dagger\big)^s \big((1-\alpha) U\ket{0}\bra{0}U^\dagger+\alpha U\ket{+}\bra{+} U^\dagger \big)^{1-s} \Big] 
        \end{align}
        Taking $U=(X+Y)/\sqrt{2}$, we obtain $f(s)=g(s)$. This follows from
        \begin{equation}
            U\ket{0}=\frac{(1+i)}{\sqrt{2}} \ket{1} \text{, }U\ket{+}=\frac{(1-i)}{\sqrt{2}} \ket{+y} \text{ and } U\ket{+y}=\frac{(1+i)}{\sqrt{2}} \ket{+}
        \end{equation}
        Now, we show that $f(s)=f(1-s)$ for all $s\in (0,1)$. Let 
        \begin{equation}
            A= (1-\alpha) \ket{0}\bra{0}+\alpha \ket{+y}\bra{+y}
        \end{equation}
        and
        \begin{equation}
            B= (1-\alpha) \ket{0}\bra{0}+\alpha \ket{+}\bra{+}
        \end{equation}
        Then, using that $\Tr(M)=\Tr(M^T)$, and $B=B^T$,
        \begin{align}
            f(s)&=\Tr((A^T)^s B^{1-s}) \nonumber \\ &=\Tr \Big [\big( (1-\alpha) \ket{0}\bra{0}+\alpha \ket{-y}\bra{-y}\big)^s \big((1-\alpha) \ket{0}\bra{0}+\alpha \ket{+}\bra{+} \big)^{1-s} \Big] \nonumber \\ &=\Tr \Big [S^\dagger S\big( (1-\alpha)  \ket{0}\bra{0}+\alpha \ket{-y}\bra{-y}\big)^s S^\dagger  S \big((1-\alpha) \ket{0}\bra{0}+\alpha \ket{+}\bra{+} \big)^{1-s} \Big] \nonumber \\ &= \Tr \Big [\big( (1-\alpha)  \ket{0}\bra{0}+\alpha \ket{+}\bra{+}\big)^s  \big((1-\alpha) \ket{0}\bra{0}+\alpha \ket{+y}\bra{+y} \big)^{1-s} \Big] \nonumber \\ &= f(1-s)
        \end{align}
        where $S=\mathrm{diag}([1,j])$ is the phase gate. Therefore, since any function of the form $s\rightarrow\Tr(\rho^s \sigma^{1-s})$ is convex for every pair of quantum states \cite{audenaert2007discriminating}, we have $\inf_{s\in (0,1)} f(s)=f(1/2)$. Consequently,
        \begin{align}\label{decomposition}
            \inf_{s\in (0,1)}\Tr((\sigma_{D}^{c_1})^s (\sigma_{D'}^{c_2})^{1-s})&=\Tr(\sqrt{A} \sqrt{B})\nonumber \\ &= \sum_{i,j\in \{-,+\}} \sqrt{\lambda_i(A) \lambda_j (B)} |\bra{v_i(A)}\ket{v_j(B)}|^2
        \end{align}
        To compute the eigenvalues and eigenvectors, we use the fact that for any state
        \begin{equation}
            \sigma=\frac{1}{2}\left( I+ r_x X+r_y Y+r_z Z \right) 
        \end{equation}
        its eigenvalues are given by $\lambda_{\pm}(\sigma)= \frac{1}{2} \left( 1\pm\|r\|_2 \right)$, where $r=[r_x,r_y,r_z]^T$, and its eigenvectors satisfy
        \begin{equation}
            \ket{v_{\pm}(\sigma)}\bra{v_{\pm}(\sigma)}=\frac{1}{2}\left( I \pm \frac{r_x X+r_y Y+r_z Z}{\|r\|_2} \right) 
        \end{equation}
        Decomposing $A$ and $B$ in the Pauli basis, we obtain $r(A)=[0,\alpha, 1-\alpha]^T$ and $r(B)=[\alpha,0,1-\alpha]^T$. Thus, both matrices have the same eigenvalues, given by $\lambda_{\pm}(A)=\lambda_{\pm}(B)= \frac{1}{2} \left( 1\pm \sqrt{\alpha^2+(1-\alpha)^2} \right)$. For simplicity, we use $w=\sqrt{\alpha^2+(1-\alpha)^2}$. Next, to compute the inner products
        \begin{align}
           |\bra{v_i(A)}\ket{v_j(B)}|^2 &=\Tr(\ket{v_{i}(A)}\bra{v_{i}(A)}\ket{v_{i}(B)}\bra{v_{i}(B)}) \nonumber \\ &=\frac{1}{2} \left( 1 + \frac{r(A)^T r(B)}{w^2}\right) \nonumber \\ &= \left\{\begin{matrix}
 \frac{1}{2}\left( 1+ \frac{(1-\alpha)^2}{w^2}\right) \text{if }i=j\\
 \frac{1}{2}\left( 1- \frac{(1-\alpha)^2}{w^2}\right) \text{if }i\neq j 
\end{matrix}\right. 
        \end{align}
        Finally, substituting into \eqref{decomposition},
        \begin{align}
            \inf_{s\in (0,1)}\Tr((\sigma_{D}^{c_1})^s (\sigma_{D'}^{c_2})^{1-s})&= \frac{1}{2}\left( 1+\frac{(1-\alpha)^2}{w^2}\right)+ \frac{\sqrt{1-w^2}}{2}\left( 1-\frac{(1-\alpha)^2}{w^2}\right)\nonumber \\ &=  \frac{1}{2}\left( 1+ \sqrt{1-w^2}+ \left( 1- \sqrt{1-w^2}\right)\frac{(1-\alpha)^2}{w^2} \right) \nonumber \\&= \frac{1}{2}\left( 1+ \sqrt{1-w^2}+ \frac{(1-\alpha)^2}{1+\sqrt{1-w^2}} \right) \nonumber \\&= \frac{1}{2}\left(  \frac{(1+ \sqrt{1-w^2})^2+(1-\alpha)^2}{1+\sqrt{1-w^2}} \right)\nonumber \\&= \frac{1}{2}\left(  \frac{ 2+2\sqrt{1-w^2}+(1-w^2)+\alpha^2-2\alpha}{1+\sqrt{1-w^2}} \right)\nonumber \\&= 1- \frac{1}{2}\left(  \frac{\alpha^2}{1+\sqrt{2\alpha (1-\alpha)}} \right) \nonumber \\ & \geq 1-\frac{\alpha^2}{2}
        \end{align}
        Consequently,
        \begin{equation}
            \Lambda_\mathcal{C}(\alpha)\leq - \ln \left( 1-\frac{\alpha^2}{2} \right)  \leq  \alpha^2
        \end{equation}
        where the second inequality uses $\ln (1+x) \geq \frac{x}{1+x}$ for $x>-1$ \cite{topsoe2007some}. 

        We now prove that this concept class is PAC-learnable using quantum examples. To this end, we first introduce some notation. Let
        \begin{align}
             &\alpha=\frac{\mathbb{P}_D(2)}{\mathbb{P}_D(2)+\mathbb{P}_D(4)} \hspace{2cm} \beta=\frac{\mathbb{P}_D(3)}{\mathbb{P}_D(3)+\mathbb{P}_D(4)}  \nonumber \\ & \tilde{\alpha}=\frac{\mathbb{P}_D(3)}{\mathbb{P}_D(1)+\mathbb{P}_D(3)} \hspace{2cm} \tilde{\beta}=\frac{\mathbb{P}_D(2)}{\mathbb{P}_D(1)+\mathbb{P}_D(2)}
        \end{align}
        Furthermore, let $\sigma_{D}^{c_1}(i)$, for $i\in{0,1}$, denote the density matrices obtained by conditioning on observing label $i$ and tracing out the corresponding subsystem, i.e.,
        \begin{equation}
            \sigma_{D}^{c_1}= \Big(\mathbb{P}_D(1)+\mathbb{P}_D(3) \Big) \sigma_{D}^{c_1}(0) \otimes \ket{0}\bra{0} +\Big(\mathbb{P}_D(2)+\mathbb{P}_D(4) \Big) \sigma_{D}^{c_1}(1) \otimes \ket{1}\bra{1}
        \end{equation}
        These states satisfy the following properties:        
        \begin{align}
            &\Tr(\sigma_D^{c_1}(1) X)= \alpha \hspace{0.6cm} \Tr(\sigma_D^{c_1}(1) Y)= 0 \hspace{0.6cm} \Tr(\sigma_D^{c_1}(0) X)= 0 \hspace{0.6cm} \Tr(\sigma_D^{c_1}(0) Y)= \tilde{\alpha} \nonumber \\&\Tr(\sigma_D^{c_2}(1) X)= 0 \hspace{0.6cm} \Tr(\sigma_D^{c_2}(1) Y)= \beta \hspace{0.6cm} \Tr(\sigma_D^{c_2}(0) X)= \tilde{\beta}  \hspace{0.6cm} \Tr(\sigma_D^{c_2}(0) Y)=0
        \end{align}
        We can now prove that the concept class is learnable. The learning algorithm is defined as follows:
        Using,
        \begin{equation}
            t= O\left(\frac{1}{\epsilon^3}\ln \frac{1}{\delta} \right)
        \end{equation}
        examples, the algorithm estimates $\hat{X}_1$, $\hat{Y}_1$, $\hat{X}_0$, and $\hat{Y}_0$, where $\hat{A}_i$ denotes the estimated expectation value of the observable $A\in\{X,Y\}$ on the reduced density matrix associated with label $i$. If the number of copies with label $i$ is smaller than $10^5\ln(2/\delta)/\epsilon^2$ or if $|\hat{X}_i-\hat{Y}_i|<\epsilon(1/2-1/50)$, then label $i$ is considered inconsistent. If label $1$ is consistent, the decision is made as follows: if $\hat{X}_1>\hat{Y}_1$, output $c_1$, otherwise, output $c_2$. If label $1$ is inconsistent, we check whether label $0$ is consistent. In this case, the decision is made as follows: if $\hat{X}_0<\hat{Y}_0$, output $c_1$, otherwise, output $c_2$. If neither label is consistent, a random answer is returned.

        The algorithm always outputs either $c_1$ or $c_2$. Therefore, without loss of generality, it suffices to consider the case where $\mathrm{dist}_{D}(c_1,c_2)>\epsilon$. Consequently, either $\mathbb{P}_D(2)> \epsilon/2$ or $\mathbb{P}_D(3)> \epsilon/2$. Consider, for instance, the case where $\mathbb{P}_D(2)> \epsilon/2$. Then, $\alpha> \epsilon/2$ and $\tilde{\beta}> \epsilon/2$. Additionally, if the ground-truth concept is $c_1$, the expected number of examples with label $1$ is $\frac{\epsilon}{2}t$. Using Hoeffding's inequality, this implies that the error of each estimate (using half of the copies for each estimate) is bounded by
        \begin{equation}
            \mu=\sqrt{\frac{8}{\epsilon t} \ln \frac{2}{\delta}}=O\left(\epsilon\right)
        \end{equation}
        with probability at least $1-\delta$. For a sufficiently large constant, $\hat{X}_1>\frac{\epsilon}{2}\pm \frac{\epsilon}{100}$ and $\hat{Y}_1=0\pm \frac{\epsilon}{100}$ with probability at least $1-\delta$. Hence, label $1$ is consistent in this case. Similarly, if the ground-truth concept is $c_2$, we obtain $\hat{X}_0>\frac{\epsilon}{2}\pm \frac{\epsilon}{100}$ and $\hat{Y}_0=0\pm \frac{\epsilon}{100}$, which implies that label $0$ is consistent. Therefore, if $\mathbb{P}_D(2)>\epsilon/2$, independently of which concept is the ground-truth, at least one label is consistent, and the algorithm makes the correct decision with high probability. The case $\mathbb{P}_D(3)>\epsilon/2$ is completely analogous.

    \end{proof}

    \section{Proof of Theorem \ref{proposition_sufficient}} \label{proof_proposition}

     In this appendix, we provide the proof of Theorem \ref{proposition_sufficient}. We first introduce an auxiliary result from \cite{chefles1998unambiguous}.

    \begin{lemma}\label{lemma_unambiguous_discrimination} 
        Let $\{\ket{\psi_i}\}_{i=1}^N$ be a set of linearly independent quantum states. Then, there exists an unambiguous measurement with the same inconclusive probability for each quantum state, satisfying $\mathbb{P}(\mathcal{I}|s)=1-\lambda_{\min}(G)$, where $G_{i,j}=\bra{\psi_i}\ket{\psi_j}$. Furthermore, among all measurements satisfying the condition that $\mathbb{P}(\mathcal{I}|s)$ is constant for all $s\in[N]$, the inconclusive probability $1-\lambda_{\min}(G)$ is optimal.
    \end{lemma}

    \begin{proof}
        As shown in \cite{chefles1998unambiguous}, any unambiguous measurement $\{\Lambda_0, \Lambda_1,\cdots, \Lambda_N\}$ must be of the form
        \begin{equation}
            \Lambda_i= p_i {\ket{\phi_i}}{\bra{\phi_i}}
        \end{equation}
        where $1-p_i$ is the probability of an inconclusive result when measuring the state $\ket{\psi_i}$, and the unnormalized vectors $\ket{\phi_i}$ satisfy
        \begin{equation}
            \bra{\phi_i}\ket{\psi_j}=\delta_{i,j}
        \end{equation}
        Using that $\Lambda_0 \succeq 0$, and $p_i=p$ for all $i\in [N]$,
        \begin{equation}
            0\preceq I-p \sum_{i=1}^{N} {\ket{\phi_i}}{\bra{\phi_i}}
        \end{equation}
        Next, we use the fact that if $A\succeq 0$, then $B^\dagger A B\succeq 0$ for any matrix $B$. In particular, the matrix $H=[\ket{\psi_1},\ket{\psi_2},\cdots, \ket{\psi_N}]$ is used. That is,
        \begin{equation}
            0\preceq H^\dagger H -p H^\dagger \left( \sum_{i=1}^{N} {\ket{\phi_i}}{\bra{\phi_i}} \right) H = G-p I 
        \end{equation}
        Therefore, the optimal value of $p$ is given by $\lambda_{\min}(G)$.
        
    \end{proof}

    We are now ready to prove Theorem \ref{proposition_sufficient}, restated here for completeness.

    \begin{theorem*}
        Let $\mathcal{X}:=\{\ket{\psi_i}\}_{i}$, where the states are linearly independent. Then any concept class $\mathcal{C}\subseteq 2^{\mathcal{X}}$ is PAC learnable from quantum examples. Moreover, its sample complexity is bounded by
        \begin{equation}
            t_{Q,\mathcal{C}}(\epsilon,\delta) \leq c_U \left(\frac{d+\ln\frac{1}{\delta}}{\epsilon \, \lambda_{\mathrm{min}}(G)}  \right) 
        \end{equation}
        where $G_{i,j}=\bra{\psi_i}\ket{\psi_j}$.
    \end{theorem*}

 \begin{proof}
        First, consider the auxiliary concept class $\mathcal{C}'$ given by:
        \begin{equation}
            \mathcal{C}'=\{c'\in 2^{[N]}: \exists\, c\in \mathcal{C} \text{ such that }c'(i)=c(\ket{\psi_i}) \,\forall i\in[N] \}
        \end{equation}
        Clearly, $\mathcal{C}'$ can be learned with
        \begin{equation}
            t=C\left(\frac{d+\ln\frac{1}{\delta}}{\epsilon}  \right) 
        \end{equation}
        queries to a classical oracle that returns pairs $(x,c'(x))$, where $x\in[N]$.
        
        Next, note that since the states are linearly independent, they can be unambiguously discriminated \cite{chefles1998unambiguous}. In particular, we are interested in measurements for which the probability of an inconclusive result is the same for all states. This condition allows us to generate classical pairs $(x,c'(x))$ by applying an unambiguous discrimination measurement to the samples $(\ket{\phi},c(\ket{\phi}))$, ensuring that the same distribution is obtained in both cases. As shown in Lemma \ref{lemma_unambiguous_discrimination}, the optimal inconclusive probability for a measurement with the same inconclusive probability for all states is given by $1-\lambda_{\min}(G)$, where $G_{i,j}=\bra{\psi_i}\ket{\psi_j}$.

        With the main idea established, we now define the algorithm:
        \begin{itemize}
            \item Using $t_{Q,\mathcal{C}}=t/(\lambda_{\mathrm{min}}(G)(1-\alpha))$ quantum pairs, we perform the optimal unambiguous discrimination measurement on each state, where $\alpha\in (0,1)$ is a constant that depends on $C$.
            \item If at least $t$ classical pairs have been obtained, then the classical algorithm that learns the concept class $\mathcal{C}'$ is used. Otherwise, a random concept is output.
        \end{itemize}
        Using the Chernoff bound (inequality 4.5 from \cite{mitzenmacher2017probability}), we can bound the probability that fewer than $t$ classical copies are generated:
        \begin{align}
            \mathbb{P} \left( \sum_{i=1}^{t_{Q,\mathcal{C}}} X_i < t \right) &\leq e^{-\frac{\alpha^2 t}{2(1-\alpha)}} \nonumber \\ &=e^{-\frac{C\alpha^2}{2(1-\alpha)} \left(\frac{d+\ln\frac{1}{\delta}}{\epsilon} \right)} \nonumber \\ & \leq e^{-\frac{C\alpha^2}{2(1-\alpha)} \left(\ln\frac{1}{\delta} \right)} = \delta^{\frac{C\alpha^2}{2(1-\alpha)}} \leq \delta
        \end{align}
        where in the last inequality we use that for any $C>0$, there exists a value of $\alpha \in (0,1)$ such that $\frac{C\alpha^2}{2(1-\alpha)}>1$.

        Hence, using 
        \begin{equation}
            t_{Q,\mathcal{C}} = \frac{t}{\lambda_{\mathrm{min}}(G)(1-\alpha)} = \frac{C}{1-\alpha} \left(\frac{d+\ln\frac{1}{\delta}}{\epsilon \, \lambda_{\mathrm{min}}(G)}  \right) 
        \end{equation}
        the algorithm outputs a hypothesis $h$ with error at most $\epsilon$ with probability at least $1-2\delta$.
        
    \end{proof}

    \section{Proof of $\lambda_{\min}(G) \leq 1-\max_{i,j: i\neq j} |\bra{\psi_i}\ket{\psi_j}|$ } \label{section_lambda}

    In this section, we prove the inequality $\lambda_{\min}(G) \leq 1 - \max_{i\neq j} |\bra{\psi_i}\ket{\psi_j}|$. To do this, we use Lemma \ref{lemma_unambiguous_discrimination} and a new lemma presented next.

    \begin{lemma}\label{auxiliary_result_lemma}
         Let $\{\ket{\psi_i}\}_{i=1}^N$ be a set of linearly independent quantum states, and let $\mathbb{P}_{\mathcal{I}}(\{\ket{\psi_i}\}_{i=1}^N)$ denote the optimal probability of an inconclusive result among methods satisfying the condition that $\mathbb{P}(\mathcal{I}|s)$ is constant. Then, for any subset $\{\ket{\psi_i}\}_{i\in S}\subseteq \{\ket{\psi_i}\}_{i=1}^N$,
        \begin{equation}
            \mathbb{P}_{\mathcal{I}}(\{\ket{\psi_i}\}_{i\in S}) \leq \mathbb{P}_{\mathcal{I}}(\{\ket{\psi_i}\}_{i=1}^N)
        \end{equation}
         
    \end{lemma}

    \begin{proof}
        The inequality follows from the fact that any method that unambiguously discriminates the set of states $\{\ket{\psi_i}\}_{i=1}^N$ with a constant probability of an inconclusive result is also a method that unambiguously discriminates any subset of $\{\ket{\psi_i}\}_{i=1}^N$ with a constant probability.
    \end{proof}

    Therefore, combining Lemma \ref{lemma_unambiguous_discrimination} and Lemma \ref{auxiliary_result_lemma}, we have that
    \begin{equation}
        1-\lambda_{\min}( G(\{\ket{\psi_i}\}_{i\in S})) \leq 1-\lambda_{\min}( G(\{\ket{\psi_i}\}_{i=1}^N))
    \end{equation}
    That is,
    \begin{equation}
        \lambda_{\min}( G(\{\ket{\psi_i}\}_{i\in S})) \geq \lambda_{\min}( G(\{\ket{\psi_i}\}_{i=1}^N))
    \end{equation}
    Finally, using the fact that for $S=\{i,j\}$, $\lambda_{\min}( G(\{\ket{\psi_i}\}_{i\in S}))= 1-|\bra{\psi_i}\ket{\psi_j}|$, it follows that,
    \begin{equation}
        1- |\bra{\psi_i}\ket{\psi_j}| \geq \lambda_{\min}( G(\{\ket{\psi_i}\}_{i=1}^N))
    \end{equation}
    Finally, since this holds for any set ${i,j}$, we can take the set that minimizes this upper bound.
\end{appendices}


\end{document}